\documentclass{article}

\usepackage{graphicx}
\usepackage{amsmath}
\usepackage{amsfonts}
\usepackage{amssymb}
\usepackage{amsthm}
\usepackage{multirow}
\usepackage{multicol}
\usepackage{bm}
\usepackage{enumitem}

\newtheorem{theorem}{Theorem}[section]
\newtheorem{lemma}[theorem]{Lemma}

\begin{document}

\title{Tight Bounds for Active Self-Assembly\\Using an Insertion Primitive\thanks{An abstract version of this work has been published as~\cite{Malchik-2014a}.}}

\author{Benjamin Hescott\footnote{Tufts University, Department of Computer Science, \texttt{hescott@cs.tufts.edu}} 
\and Caleb Malchik\footnote{Tufts University, Department of Computer Science, \texttt{caleb.malchik@tufts.edu}}  
\and Andrew Winslow\footnote{Universit\'e Libre de Bruxelles, D\'{e}partement d'Informatique, \texttt{awinslow@ulb.ac.be}}
}

\date{}

\maketitle

\begin{abstract}
We prove two limits on the behavior of a model of self-assembling particles introduced by Dabby and Chen (SODA 2013), called \emph{insertion systems}, where monomers insert themselves into the middle of a growing linear polymer.
First, we prove that the expressive power of these systems is equal to context-free grammars, answering a question posed by Dabby and Chen.
Second, we prove that systems of $k$ monomer types can deterministically construct polymers of length $n = 2^{\Theta(k^{3/2})}$ in $O(\log^{5/3}(n))$ expected time, and that this is optimal in both the number of monomer types and expected time.
\end{abstract}

\section{Introduction}

In this work we study a theoretical model of \emph{algorithmic self-assembly}, in which simple particles aggregate in a distributed manner to carry out complex functionality.
Perhaps the the most well-studied theoretical model of algorithmic self-assembly is the \emph{abstract Tile Assembly Model (aTAM)} of Winfree~\cite{Winfree-1998a} consisting of square \emph{tiles} irreversibly attach to a growing polyomino-shaped assembly according to matching edge colors.
This model is capable of Turing-universal computation~\cite{Winfree-1998a}, self-simulation~\cite{Doty-2012b}, and efficient assembly of general (scaled) shapes~\cite{Soloveichik-2007a} and squares~\cite{Adleman-2001a,Rothemund-2000a}.
Despite this power, the model is incapable of assembling shapes efficiently; a single row of $n$ tiles requires $n$ tile types and $\Omega(n^2)$ expected assembly time, and any shape with $n$ tiles requires $\Omega(\sqrt{n})$ expected time~\cite{Adleman-2001a}, even if the shape is assembled non-deterministically~\cite{Chen-2012a}.

Such a limitation may not seem so significant, except that a wide range of biological systems form complex assemblies in time polylogarithmic in the assembly size, as noted in~\cite{Dabby-2013a,Woods-2013b}.
These biological systems are capable of such growth because their particles (e.g.\ living cells) \emph{actively} carry out geometric reconfiguration.
In the interest of both understanding naturally occurring biological systems and creating synthetic systems with additional capabilities, several models of \emph{active self-assembly} have been proposed recently.
These include the graph grammars of Klavins et al.~\cite{Klavins-2004b,Klavins-2004a}, the \emph{nubots} model of Woods et al.~\cite{Chen-2014a,Chen-2013a,Woods-2013b}, and the insertion systems of Dabby and Chen~\cite{Dabby-2013a}. 
Both graph grammars and nubots are capable of a topologically rich set of assemblies and reconfigurations, but rely on stateful particles forming complex bond arrangements.
In contrast, insertion systems consist of stateless particles forming a single chain of bonds.
Indeed, all insertion systems are captured as a special case of nubots in which a linear polymer is assembled via parallel insertion-like reconfigurations, as in Theorem 5.1 of~\cite{Woods-2013a}. 
The simplicity of insertion systems makes their implementation in matter a more immediately attainable goal; Dabby and Chen~\cite{Dabby-2013b,Dabby-2013a} describe a direct implementation of these systems in DNA. 

We are careful to make a distinction between \emph{active self-assembly}, where assemblies undergo reconfiguration, and \emph{active tile self-assembly}~\cite{Gautam-2013a,Hendricks-2013a,Jonoska-2014a,Jonoska-2014b,Keenan-2013b,Majumder-2008a,Padilla-2012a,Padilla-2014a}, where tile-based assemblies change their bond structure.
Active self-assembly enables exponential assembly rates by enabling insertion of new particles throughout the assembly, while active tile self-assembly does not, since the $\Omega(\sqrt{n})$ expected-time lower bound of Chen and Doty~\cite{Chen-2012a} still applies.

\section{Definitions}

Section~\ref{sec:grammar-defns} defines standard context-free grammars, as well as a special type called \emph{symbol-pair grammars}, used in Section~\ref{sec:expressive-power}.
Section~\ref{sec:is-defns} defines insertion systems, with a small number of modifications from the definitions given in~\cite{Dabby-2013a} designed to ease readability.
Section~\ref{sec:expressive-power-defn} formalizes the notion of expressive power used in~\cite{Dabby-2013a}.

\subsection{Grammars}
\label{sec:grammar-defns}

A \emph{context-free grammar} $\mathcal{G}$ is a 4-tuple $\mathcal{G} = (\Sigma, \Gamma, \Delta, S)$.
The sets $\Sigma$ and $\Gamma$ are the \emph{terminal} and \emph{non-terminal symbols} of the grammar.
The set $\Delta$ consists of \emph{production rules} or simply \emph{rules}, each of the form $L \rightarrow R_1 R_2 \cdots R_j$ with $L \in \Gamma$ and $R_i \in \Sigma \cup \Gamma$.
Finally, the symbol $S \in \Gamma$ is a special \emph{start symbol}.
The \emph{language of $\mathcal{G}$}, denoted $L(\mathcal{G})$, is the set of finite strings that can be \emph{derived} by starting with $S$, and repeatedly replacing a non-terminal symbol found on the left-hand side of some rule in $\Delta$ with the sequence of symbols on the right-hand side of the rule.
The \emph{size} of $\mathcal{G}$ is $|\Delta|$, the number of rules in $\mathcal{G}$.
If every rule in $\Delta$ is of the form $L \rightarrow R_1 R_2$ or $L \rightarrow t$, with $R_1 R_2 \in \Gamma$ and $t \in \Sigma$, then the grammar is said to be in \emph{Chomsky normal form}.

A \emph{symbol-pair grammar}, used in Section~\ref{sec:expressive-power}, is a context-free grammar in Chomsky normal form such that each non-terminal symbol is in fact a symbol pair $(a, d)$, and each production rule has the form $(a, d) \rightarrow (a, b) (c, d)$ or $(a, d) \rightarrow t$.

\subsection{Insertion systems}
\label{sec:is-defns}

Dabby and Chen~\cite{Dabby-2013b,Dabby-2013a} describe both a physical implementation and formal model of insertion systems.
We briefly review the physical implementation, then give formal definitions.  

\textbf{Physical implementation.}
Short strands of DNA, called \emph{monomers}, are bonded via complementary base sequences to form linear sequences of monomers called \emph{polymers}.
Additional monomers are \emph{inserted} into the gap between two adjacent monomers, called an \emph{insertion site}, by bonding to the adjacent monomers and breaking the existing bond between them via a strand displacement reaction (see Figure~\ref{fig:figure}).
Each insertion then creates two new insertion sites for additional monomers to be inserted, allowing \emph{construction} of arbitrarily long polymers.

\begin{figure}[ht]
\centering
\includegraphics[scale=0.9]{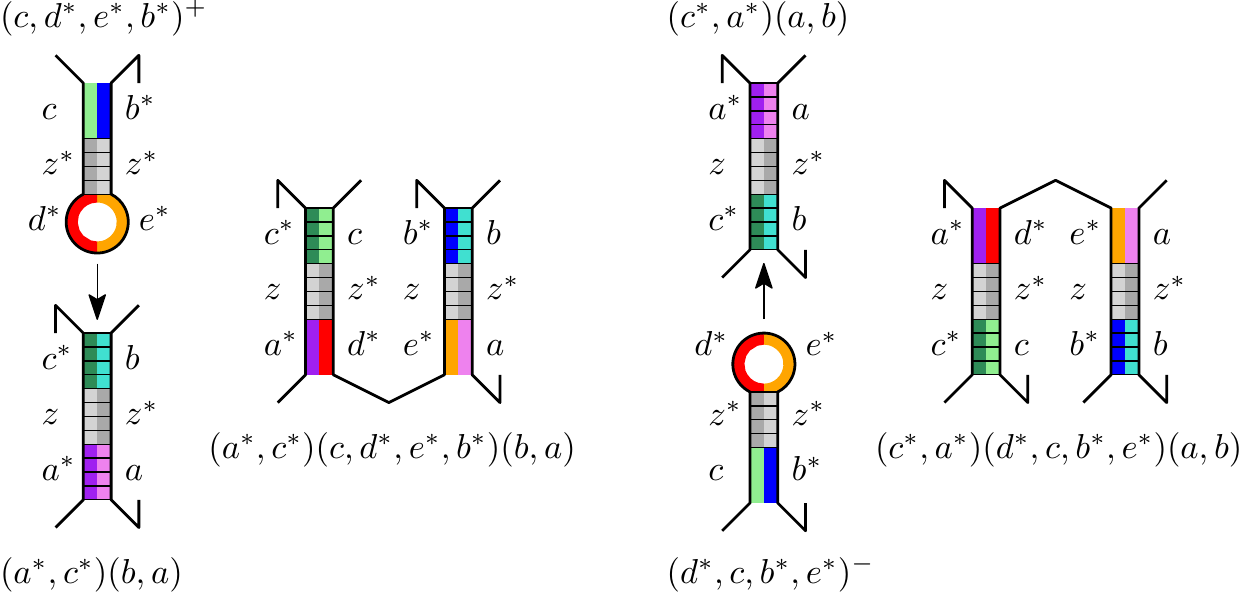}
\caption{The two types of insertions.
Each symbol denotes a DNA subsequence or its complement.
The directionality of DNA and hairpin design using generic subsequence symbols $z$, $z^*$ creates these distinct types.
This figure is loosely based on Figures~2 and~3 of~\cite{Dabby-2013a}.}
\label{fig:figure}
\end{figure}

Each monomer consists of four base sequences that form specific bonds, and only two of these can form bonds during insertion due to the monomer's hairpin design.
This design gives each insertion site or monomer one of two \emph{signs} such that a monomer can only be inserted into a site with identical sign.

\textbf{Formal model.}
An \emph{insertion system} $\mathcal{S}$ is a 4-tuple $\mathcal{S} = (\Sigma, \Delta, Q, R)$.
The first element, $\Sigma$, is a set of symbols.
Each symbol $s \in \Sigma$ has a \emph{complement} $s^*$.
We denote the complement of a symbol $s$ as $\overline{s}$, i.e. $\overline{s} = s^*$ and $\overline{s^*} = s$.

The set $\Delta$ is a set of \emph{monomer types}, each assigned a \emph{concentration}.
Each monomer is specified by a signed quadruple $(a, b, c, d)^+$ or $(a, b, c, d)^-$, where $a, b, c, d \in \Sigma \cup \{s^* : s \in \Sigma\}$, and is \emph{positive} or \emph{negative} according to its sign.
The concentration of each monomer type is a real number between~0 and~1, and the sum of call concentrations is at most~1.

The two symbols $Q = (a, b)$ and $R = (c, d)$ are special two-symbol monomers that together form the \emph{initiator} of $\mathcal{S}$.
It is required that either $\overline{a} = d$ or $\overline{b} = c$.
The \emph{size} of $\mathcal{S}$ is $|\Delta|$, the number of monomer types in $\mathcal{S}$.

A \emph{polymer} is a sequence of monomers $Q m_1 m_2 \dots m_n R$ where $m_i \in \Delta$ such that for each pair of adjacent monomers $(w, x, a, b) (c, d, y, z)$, either $\overline{a} = d$ or $\overline{b} = c$.\footnote{For readability, the signs of monomers belonging to a polymer are omitted.}
The \emph{length} of a polymer is the number of monomers it contains (including $Q$ and $R$).
The gap between every pair of adjacent monomers $(w, x, a, b) (c, d, y, z)$ in a polymer is an \emph{insertion site}, written $(a, b) (c, d)$.
Monomers can be \emph{inserted} into an insertion site $(a, b) (c, d)$ according to the following rules (seen in Figure~\ref{fig:figure}):

\begin{enumerate}
\item If $\overline{a} = d$ and $\overline{b} \neq c$, then any monomer $(\overline{b}, e, f, \overline{c})^+$ can be inserted.
\item If $\overline{a} \neq d$ and $\overline{b} = c$, then any monomer $(e, \overline{a}, \overline{d}, f)^-$ can be inserted.\footnote{In~\cite{Dabby-2013a}, this rule is described as a monomer $(\overline{d}, f, e, \overline{a})^-$ that is inserted into the polymer as $(e, \overline{a}, \overline{d}, f)$.}
\end{enumerate}

A \emph{positive} or \emph{negative} insertion site accepts only positive or negative monomers, respectively.
A \emph{dead} insertion site accepts no monomers and has the form $(a, b) (\overline{b}, \overline{a})$.
An \emph{insertion sequence} is a sequence of insertions, each specified by the site and monomer types, such that each site is created by the previous insertion.

A monomer is inserted after time $t$, where $t$ is an exponential random variable with rate equal to the concentration of the monomer type.
The set of all polymers \emph{constructed} by an insertion system is recursively defined as any polymer constructed by inserting a monomer into a polymer constructed by the system, beginning with the initiator.
Note that the insertion rules guarantee by induction that for every insertion site $(a, b) (c, d)$, either $\overline{a} = d$ or $\overline{b} = c$. 

We say that a polymer is \emph{terminal} if no monomer can be inserted into any insertion site in the polymer, and that an insertion system \emph{deterministically constructs} a polymer $P$ (i.e. is \emph{deterministic}) if every polymer constructed by the system is either $P$ or is non-terminal and has length less than that of $P$ (i.e. can become $P$).

The \emph{string representation} of a polymer is the sequence of symbols found on the polymer from left to right, e.g. $(a, b) (b^*, a, d, c) (c^*, a)$ has string representation $abb^*adcc^*a$.
We call the set of string representations of all terminal polymers of an insertion system $\mathcal{S}$ the \emph{language} of $\mathcal{S}$, denoted $L(\mathcal{S})$. 

% note: deterministic construction is not the same as deterministic derivation in a grammar.
% Instead, deterministic construction in insertion systems is equivalent to singleton languages in grammars.
% In grammars, it is clear that deterministic derivation is a stronger constraint than having a singleton language, since grammars can be ambiguous.
% In insertion systems, we do not know if deterministic construction is stronger constraint than having a unique terminal polymer.
% This is the strongest form of the ambiguity conjecture.
% It may be the case that deterministic construction (singleton language) => deterministic derivation for insertion systems.

\subsection{Expressive power}
\label{sec:expressive-power-defn}

Intuitively, a system \emph{expresses} another if the terminal polymers or strings created by the system ``look'' like the terminal polymers or strings created by the other system. 
In the simplest instance, a symbol-pair grammar $\mathcal{G}'$ is said to \emph{express} a context-free grammar $\mathcal{G}$ if $L(\mathcal{G}') = L(\mathcal{G})$.
Similarly, a grammar $\mathcal{G}$ is said to \emph{express} an insertion system $\mathcal{S}$ if $L(\mathcal{S}) = L(\mathcal{G})$, i.e. if the set of string representations of the terminal polymers of $\mathcal{S}$ equals the language of $\mathcal{G}$. 

An insertion system $\mathcal{S} = (\Sigma', \Delta', Q', R')$ is said to express a grammar $\mathcal{G} = (\Sigma, \Gamma, \Delta, S)$ if there exists a function $g : \Sigma' \cup \{s^* : s \in \Sigma'\} \rightarrow \Sigma \cup \{\varepsilon\}$ and integer $\kappa$ such that 
\begin{enumerate}
\item $\{g(s_1') g(s_2') \dots g(s_n') : s_1' s_2' \dots s_n' \in L(\mathcal{S})\} = L(\mathcal{G})$.
\item No $\kappa$ consecutive symbols of a string in $L(S)$ are mapped to $\varepsilon$ by $g$.
\end{enumerate}

The string representations of polymers have both complementary symbol and length requirements that imply they are unable to capture even simple languages, e.g. $\{aa \dots a\}$, despite intuition and claims to the contrary, e.g. Theorem 3.2 of~\cite{Dabby-2013a} that claims insertion systems express all regular languages.
Allowing $g$ to output $\varepsilon$ enables locally ``cleaning up'' string representations to eliminate complementary pairs and other debris, while $\kappa$ ensures there is a limit on the amount that can be ``swept under the rug'' locally.
A feasible stricter definition could instead use a function $g: \Delta' \rightarrow \Sigma$ (monomer types of $\mathcal{S}$ to terminal symbols of $\mathcal{S}$); it is open whether the results presented here would hold under such a definition.

\section{The Expressive Power of Insertion Systems}
\label{sec:expressive-power}

Dabby and Chen proved that any insertion system has a context-free grammar expressing it.
They construct such a grammar by creating a non-terminal for every possible insertion site 
and a production rule for every monomer type insertable into the site.
For instance, the insertion site $(a,b)(c^*,a^*)$ and monomer type $(b^*, d^*, e, c)^+$ induce non-terminal symbol $A_{(a, b)(c^*, a^*)}$ and production rule $A_{(a, b)(c^*, a^*)} \rightarrow A_{(a,b)(b^*, d^*)} A_{(e, c)(c^*,a^*)}$.
Here we give a reduction in the other direction, resolving in the affirmative the question posed by Dabby and Chen of whether context-free grammars and insertion systems have the same expressive power:

\begin{theorem}
\label{thm:IS-express-CFG}
For every context-free grammar $G$, there exists an insertion system that expresses $G$.
\end{theorem}

The primary difficulty in proving Theorem~\ref{thm:IS-express-CFG} lies in developing a way to simulate the ``complete'' replacement that occurs during derivation with the ``incomplete'' replacement that occurs at an insertion site during insertion.
For instance, $bcAbc \Rightarrow bcDDbc$ via a production rule $A \rightarrow DD$ and $A$ is completely replaced by $DD$.
On the other hand, inserting a monomer $(b^*, d, d, c)^+$ into a site $(a, b) (c^*, a^*)$ yields the consecutive sites $(a, b) (b^*, d)$ and $(d, c) (c^*, a^*)$, with $(a, b) (c^*, a^*)$ only partially replaced -- the left side of the first site and the right side of second site together form the initial site.
This behavior constrains how replacement can be captured by insertion sites, and the $\kappa$ parameter of the definition of expression (Section~\ref{sec:expressive-power-defn}) prevents eliminating the issue via additional insertions.

We overcome this difficulty by proving Theorem~\ref{thm:IS-express-CFG} in two steps. 
First, we prove that symbol-pair grammars, a constrained type of grammar with incomplete replacements, are able to express context-free grammars (Lemma~\ref{lem:PG-express-CFG}).
Second, we prove symbol-pair grammars can be expressed by insertion systems (Lemma~\ref{lem:IS-express-PG}).

\begin{lemma}
\label{lem:PG-express-CFG}
For every context-free grammar $\mathcal{G}$, there exists a symbol-pair grammar that expresses $\mathcal{G}$.
\end{lemma}

\begin{proof}
Let $\mathcal{G} = (\Sigma, \Gamma, \Delta, S)$. 
Let $n = |\Gamma|$.
Start by putting $\mathcal{G}$ into Chomsky normal form and then relabeling the non-terminals of $\mathcal{G}$ to $A_0, A_1, \dots, A_{n-1}$, with $S = A_0$.

Now we define a symbol-pair grammar $\mathcal{G}' = (\Sigma', \Gamma', \Delta', S')$ such that $L(\mathcal{G}') = L(\mathcal{G})$.
Let $\Sigma' = \Sigma$ and $\Gamma' = \{(a, d) : 0 \leq a,d < n \}$; we treat the symbols in the pairs of $\Gamma'$ as both symbols and integers.

For each production rule $A_i \rightarrow A_j A_k$ in $\Delta$, add to $\Delta'$ the set of rules $(a, d) \rightarrow (a, b) (c, d)$, with $0 \leq a < n$, $d = (i - a) \bmod n$, $b = (j - a) \bmod n$, and $c = (k - d) \bmod n$.
For each production rule $A_i \rightarrow t$ in $\Delta$, add to $\Delta'$ the set of rules $(a, d) \rightarrow t$, with $0 \leq a < n$ and $d = (i - a) \bmod n$.
Let $S' = (0, 0)$.

We claim that a partial derivation $P'$ of $\mathcal{G}'$ exists if and only if the partial derivation $P$ obtained by replacing each non-terminal $(a, d)$ in $P'$ with $A_{(a + d) \bmod n}$ is a partial derivation of $\mathcal{G}$.
By construction, a rule $(a, d) \rightarrow (a, b) (c, d)$ is in $\Delta'$ if and only if the rule $A_{(a + d) \bmod n} \rightarrow A_{(a + b) \bmod n} A_{(c + d) \bmod n}$ is in $\Delta$.
Similarly, a rule $(a, d) \rightarrow t$ is in $\Delta'$ if and only if the rule $A_{(a + d) \bmod n} \rightarrow r$ is in $\Delta$.
Also, $S' = (0, 0)$ and $S = A_{(0 + 0) \bmod n}$.
So the claim holds by induction.

Since the set of all partial derivations of $P'$ are equal to those of $P$, the completed derivations are as well and $L(\mathcal{S}') = L(\mathcal{S})$. 
So $\mathcal{G}'$ expresses $\mathcal{G}$.
\end{proof}

\begin{lemma}
\label{lem:IS-express-PG}
For every symbol-pair grammar $\mathcal{G}$, there exists an insertion system that expresses $\mathcal{G}$. 
\end{lemma}

\begin{proof}
Let $\mathcal{G} = (\Sigma, \Gamma, \Delta, S)$.
The symbol-pair grammar $\mathcal{G}$ is expressed by an insertion system $\mathcal{S} = (\Sigma', \Delta', Q', R')$ that we now define.
Let $\Sigma' = \{s_a, s_b : (a, b) \in \Gamma\} \cup \{u, x\} \cup \Sigma$.
Let $\Delta' = \Delta_1' \cup \Delta_2' \cup \Delta_3' \cup \Delta_4'$, where
\begin{align*}
\Delta_1' &= \{(s_b, u^*, s_b^*, x)^- : (a, d) \rightarrow (a, b) (c, d) \in \Delta \}\\
\Delta_2' &= \{(s_a^*, s_b, s_c^*, s_d^*)^+ : (a, d) \rightarrow (a, b) (c, d) \in \Delta \}\\
\Delta_3' &= \{(x, s_c, u, s_c)^- : (a, d) \rightarrow (a, b) (c, d) \in \Delta \}\\
\Delta_4' &= \{(s_a^*, t, x, s_d^*)^+ : (a, d) \rightarrow t \in \Delta \}
\end{align*}
Let $Q' = (u, a)$ and $R' = (b, u^*)$, where $S = (a, b)$.

For instance, the following insertions simulate applying the production rule $(0, 0) \rightarrow (0, 1) (2, 0)$ to $(0, 0)$, where $\diamond$ denotes the available insertion sites and bold the inserted monomer:
\begin{center}
$\begin{array}{c}
(u, s_0) \diamond (s_0, u^*)\\ 
(u, s_0) \diamond \bm{(s_0^*, s_1, s_2^*, s_0^*)} \diamond (s_0, u^*)\\
(u, s_0) \diamond \bm{(s_1, u^*, s_1^*, x)} (s_0^*, s_1, s_2^*, s_0^*) \diamond (s_0, u^*)\\
(u, s_0) \diamond (s_1, u^*, s_1^*, x) (s_0^*, s_1, s_2^*, s_0^*) \bm{(x, s_2, u, s_2)} \diamond (s_0, u^*)\\
(u, s_0) \diamond (s_1, u^*) \dots (u, s_2) \diamond (s_0, u^*)\\
\end{array}$
\end{center}

The subsequent application of production rules $(0, 1) \rightarrow p$ $(2, 0) \rightarrow q$ to the string $(0, 1) (2, 0)$ are simulated by the following insertions: 
\begin{center}
$\begin{array}{c}
(u, s_0) \diamond (s_1, u^*) \dots (u, s_2) \diamond (s_0, u^*)\\
(u, s_0) \bm{(s_0^*, p, x, s_1^*)} (s_1, u^*) \dots (u, s_2) \diamond (s_0, u^*)\\
(u, s_0) (s_0^*, p, x, s_1^*) (s_1, u^*) \dots (u, s_2) \bm{(s_2^*, q, x, s_0^*)} (s_0, u^*)\\
(u, s_0) (s_0^*, p, x, s_1^*) \dots (s_2^*, q, x, s_0^*) (s_0, u^*)\\
\end{array}$
\end{center}

\textbf{Insertion types.} First, it is proved that for any polymer constructed by $\mathcal{S}$, only three types of insertions of a monomer $m_2$ between two adjacent monomers $m_1 m_3$ are possible:
\begin{enumerate}
\item $m_1 \in \Delta_2'$, $m_2 \in \Delta_3'$, $m_3 \in \Delta_1'$.
\item $m_1 \in \Delta_3'$, $m_2 \in \Delta_2' \cup \Delta_4'$, $m_3 \in \Delta_1'$.
\item $m_1 \in \Delta_3'$, $m_2 \in \Delta_1'$, $m_3 \in \Delta_2'$.
\end{enumerate}
Moreover, for every adjacent $m_1 m_3$ pair satisfying one of these conditions, an insertion of some type $m_2$ from the specified set is possible.

Consider each possible combination of $m_1 \in \Delta_i'$ and $m_3 \in \Delta_j'$, respectively, with $i, j \in \{1, 2, 3, 4\}$.
Observe that for an insertion to occur at insertion site $(a, b) (c, d)$, the symbols $\overline{a}$, $\overline{b}$, $\overline{c}$, and $\overline{d}$ must each occur on some monomer.
Then since $x^*$ and $t^*$ do not appear on any monomers, any $i, j$ with $i \in \{1, 4\}$ or $j \in \{3, 4\}$ cannot occur.
This leaves monomer pairs $(\Delta_i', \Delta_j')$ with $(i, j) \in \{(2, 1), (2, 2), (3, 1), (3, 2)\}$.

Insertion sites between $(\Delta_2', \Delta_1')$ pairs have the form $(s_c^*, s_d^*) (s_d, u^*)$, so an inserted monomer must have the form $(\underline{~~}, s_c, u, \underline{~~})^-$ and is in $\Delta_3'$.
An insertion site $(s_c^*, s_d^*) (s_d, u^*)$ implies a rule of the form $(a, d) \rightarrow (a, b) (c, d)$ in $\Delta$, so there exists a monomer $(x, s_c, u, s_c^*)^- \in \Delta_3'$ that can be inserted.

Insertion sites between $(\Delta_3', \Delta_2')$ pairs have the form $(u, s_c) (s_c^*, s_b)$, so an inserted monomer must have the form $(\underline{~~}, u^*, s_b^*, \underline{~~})^-$ and thus is in $\Delta_1'$.
An insertion site $(u, s_c) (s_c^*, s_b)$ implies a rule of the form $(c, d) \rightarrow (c, b) (e, d)$ in $\Gamma$, so there exists a monomer $(s_b, u^*, s_b^*, x)^- \in \Delta_1'$ that can be inserted.

Insertion sites between $(\Delta_2', \Delta_2')$ pairs can only occur once a monomer $m_2 \in \Delta_2'$ has been inserted between a pair of adjacent monomers $m_1 m_3$ with either $m_1 \in \Delta_2'$ or $m_3 \in \Delta_2'$, but not both.
But we just proved that all such such possible insertions only permit $m_2 \in \Delta_3' \cup \Delta_1'$. 
Moreover, the initial insertion site between $Q'$ and $R'$ has the form $(u, s_a) (s_b, u^*)$ of an insertion site with $m_1 \in \Delta_3'$ and $m_3 \in \Delta_1'$.
So no pair of adjacent monomers $m_1 m_3$ are ever both from $\Delta_2'$ and no insertion site between $(\Delta_2', \Delta_2')$ pairs can ever exist.

Insertion sites between $(\Delta_3', \Delta_1')$ pairs have the form $(u, s_c) (s_b, u^*)$, so an inserted monomer must have the form $(s_c^*, \underline{~~}, \underline{~~}, s_b^*)^+$ and is in $\Delta_2'$ or $\Delta_4'$.
We prove by induction that for each such insertion site $(u, s_c) (s_b, u^*)$ that $(c, b) \in \Gamma$.
First, observe that this is true for the insertion site $(u, s_a) (s_b, u^*)$ between $Q'$ and $R'$, since $(a, b) = S \in \Gamma$. 
Next, suppose this is true for all insertion sites of some polymer and a monomer $m_2 \in \Delta_2' \cup \Delta_4'$ is about to be inserted into the polymer between monomers from $\Delta_3'$ and $\Delta_1'$.
Inserting a monomer $m_2 \in \Delta_4'$ only reduces the set of insertion sites between monomers in $\Delta_3'$ and $\Delta_1'$, and the inductive hypothesis holds.
Inserting a monomer $m_2 \in \Delta_2'$ induces new $(\Delta_3', \Delta_2')$ and $(\Delta_2', \Delta_1')$ insertion site pairs between $m_1 m_2$ and $m_2 m_3$.
These pairs must accept two monomers $m_4 \in \Delta_1$ and $m_5 \in \Delta_3$, inducing a sequence of monomers $m_1 m_4 m_2 m_5 m_3$ with adjacent pairs $(\Delta_3', \Delta_1')$, $(\Delta_1', \Delta_2')$, $(\Delta_2', \Delta_3')$, $(\Delta_3', \Delta_1')$.
Only the first and last pairs permit insertion and both are $(\Delta_3', \Delta_1')$ pairs.

Now consider the details of the three insertions yielding $m_1 m_4 m_2 m_5 m_3$, starting with $m_1 m_3$.
The initial insertion site $m_1 m_3$ must have the form $(u, s_a) (s_d, u^*)$.
So the sequence of insertions has the following form, with the last two insertions interchangeable:

\begin{center}
$\begin{array}{c}
(u, s_a) \diamond (s_d, u^*)\\
(u, s_a^*) \diamond \bm{(s_a^*, s_b, s_c^*, s_d^*)} \diamond (s_d, u^*)\\
(u, s_a) \diamond \bm{(s_b, u^*, s_b^*, x)} (s_a^*, s_b, s_c^*, s_d^*) \diamond (s_d, u^*)\\
(u, s_a) \diamond (s_b, u^*, s_b^*, x) (s_a^*, s_b, s_c^*, s_d^*) \bm{(x, s_c, u, s_c)} \diamond (s_d, u^*)\\
\end{array}$
\end{center}

Notice the two resulting $(\Delta_3', \Delta_1')$ pair insertion sites $(u, s_a) (s_b, u^*)$ and $(u, s_c) (s_d, u^*)$.
Assume, by induction, that the monomer $m_2$ must exist.
So there is a rule $(a, d) \rightarrow (a, b) (c, d) \in \Delta$ and $(a, b), (c, d) \in \Gamma$, fulfilling the inductive hypothesis.
So for every insertion site $(u, s_c) (s_b, u^*)$ between a $(\Delta_3', \Delta_1')$ pair there exists a non-terminal $(c, b) \in \Gamma$.
So for every adjacent monomer pair $m_1 m_3$ with $m_1 \in \Delta_3'$ and $m_3 \in \Delta_1'$, there exists a monomer $m_2 \in \Delta_2' \cup \Delta_4'$ that can be inserted between $m_1$ and $m_2$. 

\textbf{Partial derivations and terminal polymers.} Next, consider the sequence of insertion sites between $(\Delta_3', \Delta_1')$ pairs in a polymer constructed by a modified version of $\mathcal{S}$ lacking the monomers of $\Delta_4'$.
We claim that a polymer with a sequence $(u, s_{a_1}) (s_{b_1}, u^*), (u, s_{a_2}) (s_{b_2}, u^*), \dots, (u, s_{a_i}) (s_{b_i}, u^*)$ of $(\Delta_3', \Delta_1')$ insertion sites is constructed if and only if there is a partial derivation $(a_1, b_1) (a_2, b_2) \dots (a_i, b_i)$ of a string in $L(\mathcal{G})$.
This follows directly from the previous proof by observing that two new adjacent $(\Delta_3', \Delta_1')$ pair insertion sites $(u, s_a) (s_b, u^*)$ and $(u, s_c) (s_d, u^*)$ can replace a $(\Delta_3', \Delta_1')$ pair insertion site if and only if there exists a rule $(a, d) \rightarrow (a, b) (c, d) \in \Delta$.

Observe that any string in $L(\mathcal{G})$ can be derived by first deriving a partial derivation containing only non-terminals, then applying only rules of the form $(a, d) \rightarrow t$.
Similarly, since the monomers of $\Delta_4'$ never form half of a valid insertion site, any terminal polymer of $\mathcal{S}$ can be constructed by first generating a polymer containing only monomers in $\Delta_1' \cup \Delta_2' \cup \Delta_3'$, then only inserting monomers from $\Delta_4'$.
Also note that the types of insertions possible in $\mathcal{S}$ imply that in any terminal polymer, any triple of adjacent monomers $m_1 m_2 m_3$ with $m_1 \in \Delta_i'$, $m_2 \in \Delta_j'$, and $m_3 \in \Delta_k'$, that $(i, j, k) \in \{(4, 1, 2), (1, 2, 3), (2, 3, 4), (3, 4, 1)\}$, with the first and last monomers of the polymer in $\Delta_4'$.

\textbf{Expression.} Define the following piecewise function $g : \Sigma' \cup \{ s^* : s \in \Sigma' \} \rightarrow \Sigma \cup \{ \varepsilon \}$ that maps to $\varepsilon$ except for second symbols of monomers in $\Delta_4'$.

\begin{displaymath}
   g(s) = \left\{
     \begin{array}{ll}
       t, & \text{if } t \in \Sigma \\
       \varepsilon, & \text{otherwise}
     \end{array}
   \right.
\end{displaymath}

Observe that every string in $L(\mathcal{S})$ has length $2 + 4 \cdot (4n - 3) + 2 = 16n-8$ for some $n \geq 0$.
Also, for each string $s_1' s_2' \dots s_{16n-8}' \in L(\mathcal{S})$, $g(s_1') g(s_2') \dots g(s_{16n-8}') =  \varepsilon^3 t_1 \varepsilon^{16} t_2 \varepsilon^{16} \dots t_n \varepsilon^5$. 
There is a terminal polymer with string representation in $L(\mathcal{S})$ yielding the sequence $s_1 s_2 \dots s_n$ if and only if the polymer can be constructed by first generating a terminal polymer excluding $\Delta_4'$ monomers with a sequence of $(\Delta_3', \Delta_1')$ insertion pairs $(a_1, b_1) (a_2, b_2) \dots (a_n, b_n)$ followed by a sequence of insertions of monomers from $\Delta_4'$ with second symbols $t_1 t_2 \dots t_n$.
Such a generation is possible if and only if $(a_1, b_1) (a_2, b_2) \dots (a_n, b_n)$ is a partial derivation of a string in $L(\mathcal{G})$ and $(a_1, b_1) \rightarrow t_1, (a_2, b_2) \rightarrow t_2, \dots, (a_n, b_n) \rightarrow t_n \in \Delta$. 
So applying the function $g$ to the string representations of the terminal polymers of $\mathcal{S}$ gives $L(\mathcal{G})$, i.e. $L(\mathcal{S}) = L(\mathcal{G})$.
Moreover, the second symbol in every fourth monomer in a terminal polymer of $\mathcal{S}$ maps to a symbol of $\Sigma$ using $g$. 
So $\mathcal{S}$ expresses $\mathcal{G}$ with the function $g$ and $\kappa = 16$.
\end{proof}

\section{Positive Results for Polymer Growth}
\label{sec:positive-results}

Dabby and Chen also consider the size and speed of constructing finite polymers.
They give a construction achieving the following result:
 
\begin{theorem}[\cite{Dabby-2013a}] 
\label{thm:dabby-chen-fast}
For any positive integer $r$, there exists an insertion system with $O(r^2)$ monomer types that deterministically constructs a polymer of length $n = 2^{\Theta(r)}$ in $O(\log^3{n})$ expected time.
Moreover, the expected time has an exponentially decaying tail probability.
\end{theorem}

Here we improve on this construction significantly in both polymer length and expected running time.
In Section~\ref{sec:negative-results}, we prove that this construction is the best possible with respect to both the polymer length and construction time.

\begin{theorem}
\label{thm:types-extreme-ub}
For any positive integer $r$, there exists an insertion system with $O(r^2)$ monomer types that deterministically constructs a polymer of length~$n = 2^{\Theta(r^3)}$ in $O(\log^{5/3}(n))$ expected time.
Moreover, the expected time has an exponentially decaying tail probability.
\end{theorem}

\begin{proof}
The approach is to implement a three variable counter where each variable ranges over the values $0$ to $r$, effectively carrying out the execution of a triple for-loop. 
Insertion sites of the form $(s_a, s_b) (s_c, s_a^*)$ are used to encode the state of the counter, where $a$, $b$, and $c$ are the variables of the outer, inner, and middle loops, respectively. 
Three types of variable increments are carried out by the counter:
\begin{enumerate}[leftmargin=2cm]
\item[Inner:] If $b < r$, then $(s_a, s_b) (s_c, s_a^*) \leadsto (s_a, s_{b+1}) (s_c, s_a^*)$.
\item[Middle:] If $b = r$ and $c < r$, then $(s_a, s_b) (s_c, s_a^*) \leadsto (s_a, s_0) (s_{c+1}, s_a^*)$.
\item[Outer:] If $b = c = r$ and $a < r$, then $(s_a, s_b) (s_c, s_a^*) \leadsto (s_{a+1}, s_0) (s_0, s_{a+1}^*)$.
\end{enumerate} 

For $r = 2$, these increment types give an insertion sequence of the following form from left to right:
\begin{center}
\begin{tabular}{r@{\hskip 2pt}c@{\hskip 2pt}l@{\hskip 40pt}r@{\hskip 2pt}c@{\hskip 2pt}l@{\hskip 40pt}r@{\hskip 2pt}c@{\hskip 2pt}l}
$(s_0, s_0)$&&$(s_0, s_0^*)$&$(s_1, s_0)$&&$(s_0, s_1^*)$&$(s_2, s_0)$&&$(s_0, s_2^*)$\\[-1 pt]
&\rotatebox[origin=c]{270}{$\leadsto$}&{\tiny inner$\times 2$}&
&\rotatebox[origin=c]{270}{$\leadsto$}&{\tiny inner$\times 2$}&
&\rotatebox[origin=c]{270}{$\leadsto$}&{\tiny inner$\times 2$} \\[-1 pt]
$(s_0, s_2)$&&$(s_0, s_0^*)$&$(s_1, s_2)$&&$(s_0, s_1^*)$&$(s_2, s_2)$&&$(s_0, s_2^*)$\\[-1 pt]
&\rotatebox[origin=c]{270}{$\leadsto$}&{\tiny middle}&
&\rotatebox[origin=c]{270}{$\leadsto$}&{\tiny middle}&
&\rotatebox[origin=c]{270}{$\leadsto$}&{\tiny middle}\\[-1 pt]
$(s_0, s_0)$&&$(s_1, s_0^*)$&$(s_1, s_0)$&&$(s_1, s_1^*)$&$(s_2, s_0)$&&$(s_1, s_2^*)$\\[-1 pt]
&\rotatebox[origin=c]{270}{$\leadsto$}&{\tiny inner$\times 2$}&
&\rotatebox[origin=c]{270}{$\leadsto$}&{\tiny inner$\times 2$}&
&\rotatebox[origin=c]{270}{$\leadsto$}&{\tiny inner$\times 2$} \\[-1 pt]
$(s_0, s_2)$&&$(s_1, s_0^*)$&$(s_1, s_2)$&&$(s_1, s_1^*)$&$(s_2, s_2)$&&$(s_1, s_2^*)$\\[-1 pt]
&\rotatebox[origin=c]{270}{$\leadsto$}&{\tiny middle}&
&\rotatebox[origin=c]{270}{$\leadsto$}&{\tiny middle}&
&\rotatebox[origin=c]{270}{$\leadsto$}&{\tiny middle}\\[-1 pt]
$(s_0, s_0)$&&$(s_2, s_0^*)$&$(s_1, s_0)$&&$(s_2, s_1^*)$&$(s_2, s_0)$&&$(s_2, s_2^*)$\\[-1 pt]
&\rotatebox[origin=c]{270}{$\leadsto$}&{\tiny inner$\times 2$}&
&\rotatebox[origin=c]{270}{$\leadsto$}&{\tiny inner$\times 2$}&
&\rotatebox[origin=c]{270}{$\leadsto$}&{\tiny inner$\times 2$} \\[-1 pt]
$(s_0, s_2)$&&$(s_2, s_0^*)$&$(s_1, s_2)$&&$(s_2, s_1^*)$&$(s_2, s_2)$&&$(s_2, s_2^*)$\\[-1 pt]
&\rotatebox[origin=c]{270}{$\leadsto$}&{\tiny outer}&
&\rotatebox[origin=c]{270}{$\leadsto$}&{\tiny outer}&
&&\\[-1 pt]
$(s_1, s_0)$&&$(s_0, s_1^*)$&$(s_2, s_0)$&&$(s_0, s_2^*)$&\\[-1 pt]
\end{tabular}
\end{center}

A site is \emph{modified} by an insertion sequence that yields a new usable site where all other sites created by the insertion sequence are unusable.
For instance, we modify a site $(s_a, \bm{s_b}) (s_c, s_a^*)$ to become $(s_a, \bm{s_d}) (s_c, s_a^*)$, written $(s_a, s_b) (s_c, s_a^*) \leadsto (s_a, s_d) (s_c, s_a^*)$, by adding the monomer types $(s_b^*, x, u, s_c^*)^+$ and $(x, u^*, s_a, s_d)^-$ to the system, where $x$ is a special symbol whose complement is not found on any monomer.
These two monomer types cause the following insertion sequence, using $\diamond$ to indicate the site being modified and the inserted monomer shown in bold:
\begin{center}
$\begin{array}{c}
(s_a, s_b) \diamond (s_c, s_a^*)\\
(s_a, s_b)\bm{(s_b^*, x, u, s_c^*)} \diamond (s_c, s_a^*)\\
(s_a, s_b) (s_b^*, x, u, s_c^*) \bm{(x, u^*, s_a, s_d)} \diamond (s_c, s_a^*)
\end{array}$
\end{center}

We call this simple modification, where a single symbol in the insertion site is replaced with another symbol, a \emph{replacement}. 
There are four types of replacements, seen in Table~\ref{tab:replacements}, that can each be implemented by a pair of corresponding monomers.
\renewcommand{\arraystretch}{1.25}
\begin{table}[ht!]
\begin{center}
\begin{tabular}{| c | c |}
\hline
Replacement & Monomers \\
\hline
$(s_a, \bm{s_b}) (s_c, s_a^*) \leadsto (s_a, \bm{s_d}) (s_c, s_a^*)$ & $(s_b^*, x, u, s_c^*)^+$, $(x, u^*, s_a, s_d)^-$   \\
$(s_a, s_b) (\bm{s_c}, s_a^*) \leadsto (s_a, s_b) (\bm{s_d}, s_a^*)$ & $(s_b^*, u, x, s_c^*)^+$, $(s_d, s_a^*, u^*, x)^-$ \\
$(\bm{s_b}, s_a) (s_a^*, s_c) \leadsto (\bm{s_d}, s_a) (s_a^*, s_c)$ & $(x, s_b^*, s_c^*, u)^-$, $(u^*, x, s_d, s_a)^+$   \\
$(s_b, s_a) (s_a^*, \bm{s_c}) \leadsto (s_b, s_a) (s_a^*, \bm{s_d})$ & $(u, s_b^*, s_c^*, x)^-$, $(s_a^*, s_d, x, u^*)^+$ \\ 
\hline
\end{tabular}
\end{center}
\caption{The four types of replacement steps and monomer pairs that implement them.
The symbol $u$ can be any symbol, and $x$ is a special symbol whose complement does not appear on any monomer.}
\label{tab:replacements}
\end{table}

Each of the three increment types are implemented using a sequence of site modifications.
The resulting triple for-loop carries out a sequence of $\Theta(r^3)$ insertions to construct a $\Theta(r^3)$-length polymer.
A $2^{\Theta(r^3)}$-length polymer is achieved by simultaneously duplicating each site during each inner increment.
In the remainder of the proof, we detail the implementation of each increment type, starting with the simplest: middle increments.

\textbf{Middle increment.}
A middle increment of a site $(s_a, s_b) (s_c, s_a^*)$ occurs when the site has the form $(s_a, s_r) (s_c, s_a^*)$ with $0 \leq c < r$, performing the modification $(s_a, s_r) (s_c, s_a^*) \leadsto (s_a, s_0) (s_{c+1}, s_a^*)$.
We implement middle increments using a sequence of three replacements:
$$ (s_a, s_r) (s_c, s_a^*) \overset{1}{\leadsto} (s_a, s_r) (s_{f_1(c)}, s_a^*) \overset{2}{\leadsto} (s_a, s_0) (s_{f_1(c)}, s_a^*) \overset{3}{\leadsto} (s_a, s_0) (s_{c+1}, s_a^*) $$

where $f_i(n) = n + 2ir^2$.
Use of the function $f$ avoids unintended interactions between monomers, since for any $n_1, n_2 \in \{0, 1, \dots, r\}$, $f_i(n_1) \neq f_j(n_2)$ for all $i \neq j$.
Compiling this sequence of replacements into monomer types gives the following monomers:

\begin{enumerate}[label=Step \arabic*:, leftmargin=2cm]
\item $(s_r^*, s_{f_2(c)}, x, s_c^*)^+$ and $(s_{f_1(c)}, s_a^*, s_{f_2(c)}^*, x)^-$.
\item $(s_r^*, x, s_{f_3(c)}, s_{f_1(c)}^*)^+$ and $(x, s_{f_3(c)}^*, s_a, s_0)^-$.
\item $(s_0^*, s_{f_4(c+1)}, x, s_{f_1(c)}^*)^+$ and $(s_{c+1}, s_a^*, s_{f_4(c+1)}^*, x)^-$. 
\end{enumerate}

This set of monomers results in the following sequence of insertions:
\begin{center}
$\begin{array}{c}
(s_a, s_r) \diamond (s_c, s_a^*)\\
(s_a, s_r) \diamond \bm{(s_r^*, s_{f_2(c)}, x, s_c^*)} (s_c, s_a^*)\\
(s_a, s_r) \diamond \bm{(s_{f_1(c)}, s_a^*, s_{f_2(c)}^*, x)} (s_r^*, s_{f_2(c)}, x, s_c^*) (s_c, s_a^*)\\
(s_a, s_r) \diamond (s_{f_1(c)}, s_a^*)\\
(s_a, s_r) \bm{(s_r^*, x, s_{f_3(c)}, s_{f_1(c)}^*)} \diamond (s_{f_1(c)}, s_a^*)\\
(s_a, s_r) (s_r^*, x, s_{f_3(c)}, s_{f_1(c)}^*) \bm{(x, s_{f_3(c)}^*, s_a, s_0)} \diamond (s_{f_1(c)}, s_a^*)\\
(s_a, s_0) \diamond (s_{f_1(c)}, s_a^*)\\
(s_a, s_0) \diamond \bm{(s_0^*, s_{f_4(c+1)}, x, s_{f_1(c)}^*)} (s_{f_1(c)}, s_a^*)\\
(s_a, s_0) \diamond \bm{(s_{c+1}, s_a^*, s_{f_4(c+1)}^*, x)} (s_0^*, s_{f_4(c+1)}, x, s_{f_1(c)}^*) (s_{f_1(c)}, s_a^*) \\
(s_a, s_0) \diamond (s_{c+1}, s_a^*)\\
\end{array}$
\end{center}

Since each inserted monomer has an instance of $x$, all other insertion sites created are unusable.
This is true of the insertions used for outer increments and duplications as well.

\textbf{Outer increment.}
An outer increment of the site $(s_a, s_b) (s_c, s_a^*)$ occurs when the site has the form $(s_a, s_r) (s_r, s_a^*)$ with $0 \leq a < r$.
We implement this step using a four-step sequence of three normal replacements and a special quadruple replacement (Step~2):
\begin{center}
$\begin{array}{c}
(s_a, s_r) (s_r, s_a^*) \overset{1}{\leadsto} (s_a, s_{f_6(a)}^*) (s_r, s_a^*) \overset{2}{\leadsto} (s_{a+1}, s_{f_7(r)}) (s_{f_6(a)}, s_{a+1}^*)\\
(s_{a+1}, s_{f_7(r)}) (s_{f_6(a)}, s_{a+1}^*) \overset{3}{\leadsto} (s_{a+1}, s_0) (s_{f_6(a)}, s_{a+1}^*) \overset{4}{\leadsto} (s_{a+1}, s_0) (s_0, s_{a+1}^*)
\end{array}$
\end{center}

As with middle increments, we compile replacement steps~1,~2, and~4 into monomers using Table~\ref{tab:replacements}:
\begin{enumerate}[label=Step \arabic*:, leftmargin=2cm]
\item $(s_r^*, x, s_{f_5(r)}, s_r^*)^+$ and $(x, s_{f_5(r)}^*, s_a, s_{f_6(a)}^*)^-$.
\item $(s_{f_6(a)}, s_{a+1}^*, x, s_r^*)^+$ and $(x, s_a^*, s_{a+1}, s_{f_7(r)})^-$.
\item $(s_{f_7(r)}^*, x, s_{f_8(r)}, s_{f_6(a)}^*)^+$ and $(x, s_{f_8(r)}^*, s_{a+1}, s_0)^-$.
\item $(s_0^*, s_{f_9(a)}, x, s_{f_6(a)}^*)^+$ and $(s_0, s_{a+1}^*, s_{f_9(a)}^*, x)^-$.
\end{enumerate} 

Here is the sequence of insertions, using $\diamond$ to indicate the site being modified and the inserted monomer shown in bold:
\begin{center}
$\begin{array}{c}
(s_a, s_r) \diamond (s_r, s_a^*) \\ 
(s_a, s_r) \bm{(s_r^*, x, s_{f_5(r)}, s_r^*)} \diamond (s_r, s_a^*) \\ 
(s_a, s_r) (s_r^*, x, s_{f_5(r)}, s_r^*) \bm{(x, s_{f_5(r)}^*, s_a, s_{f_6(a)}^*)} \diamond (s_r, s_a^*) \\ 
(s_a, s_{f_6(a)}^*) \diamond (s_r, s_a^*) \\
(s_a, s_{f_6(a)}^*) \diamond \bm{(s_{f_6(a)}, s_{a+1}^*, x, s_r^*)} (s_r, s_a^*) \\ 
(s_a, s_{f_6(a)}^*) \bm{(x, s_a^*, s_{a+1}, s_{f_7(r)})} \diamond (s_{f_6(a)}, s_{a+1}^*, x, s_r^*) (s_r, s_a^*) \\
(s_{a+1}, s_{f_7(r)}) \diamond (s_{f_6(a)}, s_{a+1}^*) \\
(s_{a+1}, s_{f_7(r)}) \bm{(s_{f_7(r)}^*, x, s_{f_8(r)}, s_{f_6(a)}^*)} \diamond (s_{f_6(a)}, s_{a+1}^*) \\
(s_{a+1}, s_{f_7(r)}) (s_{f_7(r)}^*, x, s_{f_8(r)}, s_{f_6(a)}^*) \bm{(x, s_{f_8(r)}^*, s_{a+1}, s_0)} \diamond (s_{f_6(a)}, s_{a+1}^*) \\
(s_{a+1}, s_0) \diamond (s_{f_6(a)}, s_{a+1}^*) \\
(s_{a+1}, s_0) \diamond \bm{(s_0^*, s_{f_9(a)}, x, s_{f_6(a)}^*)} (s_{f_6(a)}, s_{a+1}^*) \\
(s_{a+1}, s_0) \diamond \bm{(s_0, s_{a+1}^*, s_{f_9(a)}^*, x)} (s_0^*, s_{f_9(a)}, x, s_{f_6(a)}^*) (s_{f_6(a)}, s_{a+1}^*) \\
(s_{a+1}, s_0) \diamond (s_0, s_{a+1}^*) \\
\end{array}$
\end{center}

\textbf{Inner increment.}
The inner increment has two phases.
The first phase (Steps~1-2) performs duplication, modifying the initial site to a pair of sites: $(s_a, s_b) (s_c, s_a^*) \leadsto (s_a, s_b) (s_{f_{10}(c)}, s_a^*) \dots (s_a, s_{b+1}) (s_c, s_a^*)$, yielding an incremented version of the original site and one other site.
The second phase (Steps~3-5) is $(s_a, s_b) (s_{f_{10}(c)}, s_a^*) \leadsto (s_a, s_{b+1}) (s_c, a^*)$, transforming the second site into an incremented version of the original site.

For the first phase, we use the three monomers:

\begin{enumerate}[leftmargin=2cm]
\item[Step 1:] $(s_b^*, s_{f_{10}(c)}, s_{f_{10}(b+1)}, s_c^*)^+$.
\item[Step 2:] $(s_{f_{11}(c)}, s_a^*, s_{f_{10}(c)}^*, x)^-$ and $(x, s_{f_{10}(b+1)}^*, s_a, s_{b+1})^-$.
\end{enumerate}

The resulting sequence of insertions is
\begin{center}
$\begin{array}{c}
(s_a, s_b) \diamond (s_c, s_a^*) \\
(s_a, s_b) \diamond \bm{(s_b^*, s_{f_{10}(c)}, s_{f_{10}(b+1)}, s_c^*)} \diamond (s_c, s_a^*) \\
(s_a, s_b) \diamond \bm{(s_{f_{11}(c)}, s_a^*, s_{f_{10}(c)}^*, x)} (s_b^*, s_{f_{10}(c)}, s_{f_{10}(b+1)}, s_c^*) \diamond (s_c, s_a^*) \\
(s_a, s_b) \diamond (s_{f_{11}(c)}, s_a^*, s_{f_{10}(c)}^*, x) (s_b^*, s_{f_{10}(c)}, s_{f_{10}(b+1)}, s_c^*) \bm{(x, s_{f_{10}(b+1)}^*, s_a, s_{b+1})} \diamond (s_c, s_a^*) \\
(s_a, s_b) \diamond (s_{f_{11}(c)}, s_a^*) \dots (s_a, s_{b+1}) \diamond (s_c, s_a^*) \\
\end{array}$
\end{center}

The last two insertions occur independently and may happen in the opposite order of the sequence depicted here.
In the second phase, the site $(s_a, s_b) (s_{f_{11}(c)}, s_a^*)$ is transformed into $(s_a, s_{b+1}) (s_c, s_a^*)$ by a sequence of replacement steps:
\begin{center}
$\begin{array}{c}
(s_a, s_b) (s_{f_{11}(c)}, s_a^*) \overset{3}{\leadsto} (s_a, s_{f_{12}(b)}) (s_{f_{11}(c)}, s_a^*) \overset{4}{\leadsto} (s_a, s_{f_{12}(b)}) (s_c, s_a^*) \overset{5}{\leadsto} (s_a, s_{b+1}) (s_c, s_a^*)\\
\end{array}$
\end{center} 

As with previous sequences of replacement steps, we compile this sequence into a set of monomers:
\begin{enumerate}[leftmargin=2cm]
\item[Step 3:] $(s_b^*, x, s_{f_{13}(b)}, s_{f_{11}(c)}^*)^+$ and $(x, s_{f_{13}(b)}^*, s_a, s_{f_{12}(b)})^-$. 
\item[Step 4:] $(s_{f_{12}(b)}^*, s_{f_{14}(c)}, x, s_{f_{11}(c)}^*)^+$ and $(s_c, s_a^*, s_{f_{14}(c)}^*, x)^-$. 
\item[Step 5:] $(s_{f_{12}(b)}^*, x, s_{f_{15}(b+1)}, s_c^*)^+$ and $(x, s_{f_{15}(b+1)}^*, s_a, s_{b+1})^-$.
\end{enumerate}

The resulting sequence of insertions is
\begin{center}
$\begin{array}{c}
(s_a, s_b) \diamond (s_{f_{11}(c)}, s_a^*) \\
(s_a, s_b) \bm{(s_b^*, x, s_{f_{13}(b)}, s_{f_{11}(c)}^*)} \diamond (s_{f_{11}(c)}, s_a^*) \\
(s_a, s_b) (s_b^*, x, s_{f_{13}(b)}, s_{f_{11}(c)}^*) \bm{(x, s_{f_{13}(b)}^*, s_a, s_{f_{12}(b)})} \diamond (s_{f_{11}(c)}, s_a^*) \\
(s_a, s_{f_{12}(b)}) \diamond (s_{f_{11}(c)}, s_a^*) \\
(s_a, s_{f_{12}(b)}) \diamond \bm{(s_{f_{12}(b)}^*, s_{f_{14}(c)}, x, s_{f_{11}(c)}^*)} (s_{f_{11}(c)}, s_a^*) \\
(s_a, s_{f_{12}(b)}) \diamond \bm{(s_c, s_a^*, s_{f_{14}(c)}^*, x)} (s_{f_{12}(b)}^*, s_{f_{14}(c)}, x, s_{f_{11}(c)}^*) (s_{f_{11}(c)}, s_a^*) \\
(s_a, s_{f_{12}(b)}) \diamond (s_c, s_a^*) \\
(s_a, s_{f_{12}(b)}) \bm{(s_{f_{12}(b)}^*, x, s_{f_{15}(b+1)}, s_c^*)} \diamond (s_c, s_a^*) \\
(s_a, s_{f_{12}(b)}) (s_{f_{12}(b)}^*, x, s_{f_{15}(b+1)}, s_c^*) \bm{(x, s_{f_{15}(b+1)}^*, s_a, s_{b+1})} \diamond (s_c, s_a^*) \\
(s_a, s_{b+1}) \diamond (s_c, s_a^*) \\
\end{array}$
\end{center}

When combined, the two phases of duplication modify $(s_a, s_b) (s_c, s_a^*)$ to become $(s_a, s_{b+1}) (s_c, s_a^*) \dots (s_a, s_{b+1}) (s_c, s_a^*)$, where all sites between the duplicated sites are unusable.
Notice that although we need to duplicate $\Theta(r^3)$ distinct sites, only $\Theta(r^2)$ monomers are used in the implementation since each monomer either does not depend on $a$, e.g. $(s_b^*, x, s_{f_{13}(b)}, s_{f_{11}(c)}^*)^+$, or does not depend on $c$, e.g. $(x, s_{f_{13}(b)}^*, s_a, s_{f_{12}(b)})^-$.

\renewcommand{\arraystretch}{1.25}

\begin{table}
\begin{center}
\begin{tabular}{| c | l l |}
\hline
Step     & \multicolumn{2}{|c|}{Inner monomer types ($b < r$)}                        \\
\hline
1        & \multicolumn{2}{|c|}{$(s_b^*, s_{f_{10}(c)}, s_{f_{10}(b+1)}, s_c^*)^+$}                                \\ 
2        & $(s_{f_{11}(c)}, s_a^*, s_{f_{10}(c)}^*, x)^-$           & $(x, s_{f_{10}(b+1)}^*, s_a, s_{b+1})^-$     \\ 
3        & $(s_b^*, x, s_{f_{13}(b)}, s_{f_{11}(c)}^*)^+$           & $(x, s_{f_{13}(b)}^*, s_a, s_{f_{12}(b)})^-$ \\
4        & $(s_{f_{12}(b)}^*, s_{f_{14}(c)}, x, s_{f_{11}(c)}^*)^+$ & $(s_c, s_a^*, s_{f_{14}(c)}^*, x)^-$         \\
5        & $(s_{f_{12}(b)}^*, x, s_{f_{15}(b+1)}, s_c^*)^+$         & $(x, s_{f_{15}(b+1)}^*, s_a, s_{b+1})^-$     \\ [3pt]
\hline
Step     & \multicolumn{2}{|c|}{Middle monomer types ($c < r$)} \\
\hline
1        & $(s_r^*, s_{f_2(c)}, x, s_c^*)^+$          & $(s_{f_1(c)}, s_a^*, s_{f_2(c)}^*, x)^-$  \\
2        & $(s_r^*, x, s_{f_3(c)}, s_{f_1(c)}^*)^+$   & $(x, s_{f_3(c)}^*, s_a, s_0)^-$           \\
3        & $(s_0^*, s_{f_4(c+1)}, x, s_{f_1(c)}^*)^+$ & $(s_{c+1}, s_a^*, s_{f_4(c+1)}^*, x)^-$   \\ [3pt]
\hline
Step     & \multicolumn{2}{|c|}{Outer monomer types ($a < r$)} \\
\hline
1        & $(s_r^*, x, s_{f_5(r)}, s_r^*)^+$               & $(x, s_{f_5(r)}^*, s_a, s_{f_6(a)}^*)^-$   \\
2        & $(s_{f_6(a)}, s_{a+1}^*, x, s_r^*)^+$           & $(x, s_a^*, s_{a+1}, s_{f_7(r)})^-$        \\
3        & $(s_{f_7(r)}^*, x, s_{f_8(r)}, s_{f_6(a)}^*)^+$ & $(x, s_{f_8(r)}^*, s_{a+1}, s_0)^-$        \\
4        & $(s_0^*, s_{f_9(a)}, x, s_{f_6(a)}^*)^+$        & $(s_0, s_{a+1}^*, s_{f_9(a)}^*, x)^-$      \\ [3pt]
\hline
\end{tabular}
\end{center}
\caption{The set of all monomer types used to deterministically construct a monomer of size $2^{\Theta(r^3)}$ using $O(r^2)$ monomer types.}
\label{tab:all-monomers-types-extreme-ub}
\end{table}

\textbf{Putting it together.}
The system starts with the intiator $(s_0, s_0) (s_0, s_0^*)$.
Each increment of the counter occurs either through a middle increment, outer increment, or a duplication.
The total set of monomers is seen in Table~\ref{tab:all-monomers-types-extreme-ub}.
There are at most $(r+1)^2$ monomer types in each family (each row of Table~\ref{tab:all-monomers-types-extreme-ub}) and $O(r^2)$ monomer types total.

The system is deterministic if no pair of monomers can be inserted into any insertion site appearing during construction.
It can be verified by an inspection of Table~\ref{tab:all-monomers-types-extreme-ub} that any two positive monomers have distinct pairs of first and fourth symbols, and any pair of negative monomers have distinct pairs of second and third symbols.
So no two monomers can be inserted into the same site and thus the system is deterministic.

The size $P_i$ of a subpolymer with an initiator encoding some value $i$ between $0$ and $(r+1)^3-1$ can be bounded by $2P_{i+2} + 9 \leq P_i \leq 2P_{i+1} + 9$, since either $i+1$ or $i+2$ is an inner increment step and no step inserts more than 9 monomers. 
Moreover, $P_{(r+1)^3-2} \geq 1$.
So $P_0 + 2$, the size of the terminal polymer, is $2^{\Theta(r^3)}$. 
  
\textbf{Running time.}
Define the concentration of each monomer type to be equal.
There are $12r^2 + 24r + 3 \leq 39r^2$ monomer types, so each monomer type has concentration at least $1/(39r^2)$.
The polymer is complete as soon as every counter's variables have reached the value $a = b = c = r$, i.e. every site encoding a counter has been modified to become $(s_r, s_r) (s_r, s_r^*)$ and the monomer $(s_r^*, x, s_{f_5(r)}, s_r^*)^+$ has been inserted.

There are fewer than $2^{r^3}$ such insertions, and each insertion requires at most $9 \cdot (r+1)^3 \leq 72r^3$ previous insertions to occur.
So an upper bound on the expected time $T_r$ for each such insertion is described as a sum of $72r^3$ random variables, each with expected time $39r^2$.
The Chernoff bound for independent exponential random variables~\cite{Chernoff-1952} implies the following upper bound on $T_r$:

\begin{align*}
{\rm Prob}[T_r > 39r^2 \cdot 72r^3(1 + \delta)] &\leq e^{-39 \cdot 72r^5 \delta^2 / (2 + \delta)} \\
&\leq e^{-r^5 \delta^2 / (2 + \delta)} \\
&\leq e^{-r^5 \delta^2 / (2\delta)} \rm{~for~all~} \delta \geq 2 \\
&\leq e^{-r^5 \delta / 2} \\
\end{align*}

Let $T_{\mathcal{S}_r}$ be the total running time of the system. 
Then we can bound $T_{\mathcal{S}_r}$ from above using the bound for $T_r$:

\begin{align*}
{\rm Prob}[T_{\mathcal{S}_r} > 39r^2 \cdot 72r^3(1 + \delta)] &\leq 2^{r^3} \cdot e^{-r^5 \delta / 2} \\
&\leq 2^{r^3} 2^{-r^5 \delta/2} \\ 
&\leq 2^{r^3 - r^5 \delta/2} \\
&\leq 2^{r^5 \delta/4 - r^5 \delta/2} \text{~for~all~} \delta \geq 4 \\
&\leq 2^{-r^5\delta/4}
\end{align*}

So ${\rm Prob}[T_{\mathcal{S}_r} > 39r^2 \cdot 72r^3(1 + \delta)] \leq 2^{-r^5\delta/4}$ for all $\delta \geq 4$.
So the expected value of $T_{\mathcal{S}_r}$, the construction time, is $O(r^5) = O(\log^{5/3}(n))$ with an exponentially decaying tail probability.
\end{proof}

\section{Negative Results for Polymer Growth}
\label{sec:negative-results}

Here we show that the construction in the previous section is the best possible.
We start by proving a helpful lemma on the number of insertion sites that accept at least one monomer type, which we call \emph{usable} insertion sites.

\begin{lemma}
\label{lem:usable-insertion-sites-ub}
Any insertion system with $k$ monomer types has at most $4k^{3/2}$ usable insertion sites.
\end{lemma}

\begin{proof}
Let $\mathcal{S} = (\Sigma, \Delta, Q, R)$ be an insertion system that deterministically constructs a polymer of length $n$.
Let $k = |\Delta|$ (the number of monomer types in $\mathcal{S}$), and relabel the symbols in $\Sigma \cup \{s^* : s \in \Sigma\}$ as $s_1, s_2, \dots, s_{4k}$, with some of these symbols possibly unused.
Define the sets
$L_i = \{ (s_a, s_b, s_i, s_c)^{\pm} \in \Delta \}$ and
$R_i = \{ (s_a, s_i, s_b, s_c)^{\pm} \in \Delta \}$.
We will consider the number of usable insertion sites of $\mathcal{S}$, and define $U_i = \{ (s_i, s_b) (s_c, \overline{s_i}) \rm{~is~usable} \}$.

Since each monomer type can only be inserted into one site in each $U_i$, $|U_i| \leq k$, and since each usable site requires a distinct pair of right and left monomer pairs, $|U_i| \leq |L_i| \cdot |R_i|$.
So $|U_i| = \rm{min}(k, |L_i| \cdot |R_i|)$.
Since each monomer type appears in exactly one $L_i$ and $R_i$, $\sum_{i=1}^{4k}{|L_i|} = \sum_{i=1}^{4k}{|R_i|} = k$.

Consider maximizing $\sum_{i=1}^{4k}|U_i| = \sum_{i=1}^{4k}\rm{min}(k, |L_i| \cdot |R_i|)$ subject to $\sum_{i=1}^{4k}{|L_i|} = \sum_{i=1}^{4k}{|R_i|} = k$.
Clearly $|L_i| \cdot |R_i| \leq \rm{max}(|L_i|, |R_i|)^2$, and if we define $B_i = L_i \cup R_i$, $|L_i| \cdot |R_i| \leq |B_i|^2$.
Then $\sum_{i=1}^{4k}|U_i| \leq \sum_{i=1}^{4k}|B_i|^2$ with $\sum_{i=1}^{4k}|B_i| = 2k$ and $|B_i| \leq \sqrt{k}$.
So $\sum_{i=1}^{4k}|U_i| \leq (\sqrt{k})^2 \cdot 2\sqrt{k}$ and thus $\sum_{i=1}^{4k}|U_i| \leq 2k^{3/2}$.
So the set of all usable sites of the form $(s_i, s_b) (s_c, \overline{s_i})$ has size $2k^{3/2}$.

A similar argument using the monomer sets
$L_i' = \{ (s_a, s_b, s_c, s_i)^{\pm} \in \Delta \}$,
$R_i' = \{ (s_i, s_a, s_b, s_c)^{\pm} \in \Delta \}$, and insertion site set
$U_i' = \{ (s_b, s_i) (\overline{s_i}, s_c) \rm{~is~usable} \}$ suffices to prove that the set of all usable sites of the form $(s_b, s_i) (\overline{s_i}, s_c)$ also has size $2k^{3/2}$.
Since these describe all usable sites, $\mathcal{S}$ has at most $4k^{3/2}$ total usable sites.
\end{proof}

\begin{theorem}
\label{thm:monomer-types-lb}
Any polymer deterministically constructed by an insertion system with $k$ monomer types has length $2^{O(k^{3/2})}$.
\end{theorem}

\begin{proof}
Let $\mathcal{S}$ be a system with $k$ monomer types that deterministically constructs a polymer.
By Lemma~\ref{lem:usable-insertion-sites-ub}, $\mathcal{S}$ has $O(k^{3/2})$ usable sites.
As observed by Dabby and Chen, $\mathcal{S}$ can be expressed by a grammar $\mathcal{G}_{\mathcal{S}}$ with at most $4k^{3/2}$ non-terminal symbols, where each insertion site $(a, b) (c, d)$ corresponds to a non-terminal $A_{a,b,c,d}$, and each monomer type $(e, f, g, h)^{\pm}$ insertable into the site corresponds to a rule $A_{a, b, c, d} \rightarrow A_{a, b, e, f} A_{g, h, c, d}$.

Let $\sigma$ be a string in $L(\mathcal{G}_{\mathcal{S}})$ of length $n$.
So the (binary) derivation tree of any derivation of $\sigma$ contains a path of length at least $\log_2{n}$.
If $\log_2{n} > 4k^{3/2}$, then this path must contain at least two occurrances of the same non-terminal symbol.
The portion of the path between these two occurrances can be pumped to derive strings of arbitrary lengths, so $L(\mathcal{G}_{\mathcal{S}})$ is infinite.
So $L(\mathcal{S}) \neq L(\mathcal{G}_{\mathcal{S}})$ and $\mathcal{G}_{\mathcal{S}}$ does not express $\mathcal{S}$, a contradiction.
Thus $\log_2{n} \leq 4k^{3/2}$ for every string in $L(\mathcal{G}_{\mathcal{S}})$ and the length of the polymer deterministically constructed by $\mathcal{S}$ is $2^{O(k^{3/2})}$.
\end{proof}

\begin{theorem}
\label{thm:deterministic-growth-speed-lb}
Deterministically constructing a polymer of length $n$ takes $\Omega(\log^{5/3}(n))$ expected time.
\end{theorem}

\begin{proof}
The proof approach is to prove a lower bound on the expected time to carry out an insertion sequence of length $\Omega(\log{n})$ involving (by Lemma~\ref{lem:usable-insertion-sites-ub}), $\Omega(\log{n})$ distinct monomer types.
This is converted into a minimization problem for the expected time, whose optimal solutions shown algebraically to be $\Omega(\log^{5/3}(n))$.

\textbf{A long insertion sequence.}
Since each insertion only increases the number of insertion sites by one, the system must carry out an insertion sequence of length at least $\log_2{n}$ when constructing the polymer.
No insertion site appears twice in this sequence, since otherwise the system (non-deterministically) constructs polymers of arbitrary length.

Suppose, for the sake of contradiction, that an insertion site in the sequence accepts monomer types $m_1$ and $m_2$, and inserts $m_1$ into some polymer.
Then all polymers constructed by the system without $m_1$ and, seperately, the system without $m_1$ are constructed by the system and each has polymers not constructed by the other.
So the system cannot deterministically construct a polymer, a contradiction, and so no insertion site in the sequence accepts more than one monomer type.

Thus the $\log_2{n}$ (or more) distinct insertion sites appearing in the insertion sequence each accept a unique monomer type.
The remainder of the proof is develop a lower bound for the total expected time of the insertions in this sequence.

\textbf{An optimization problem.}
By linearity of expectation, the total expected time of the insertions is equal to the sum of the expected time for each insertion.
Because each insertion site accepts a unique monomer type, the expected time to carry out the insertion is equal to the reciprocal of concentration of this type.
Let $k$ be the number of monomer types inserted into the sites in the subsequence.
Let $c_1, c_2, \dots, c_k$ be the sums of the concentrations of these types,
and $x_1, x_2, \dots, x_k$ be the number of times a monomer from each part is inserted during the subsequence.
Then the total expected time for all of the insertions in the subsequence is $\sum_{i=1}^{k}x_i/c_i$.
Moreover, these variables are subject to the following constraints:

\begin{enumerate} \itemsep5pt
\item $\sum_{i=1}^{k}x_i \geq \log_2{n}/2$ (total number of insertions is at least $\log_2{n}/2$).
\item $\sum_{i=1}^{k}c_i \leq 1$ (total concentration is at most 1).
\item $k \geq \log^{2/3}(n)/4$ (monomer types is at least $\log^{2/3}(n)/4$, Lemma~\ref{lem:usable-insertion-sites-ub}).
\end{enumerate}

\textbf{Minimizing expected time.}
Consider minimizing the total expected time subject to these constraints, starting with proving that $x_i/c_i = x_j/c_j$ for all $1 \leq i, j \leq k$.
That is, that the ratio of the number of times a monomer type is inserted in the subsequence to the type's concentration is equal for all types.
Assume, without loss of generality, that $x_i/c_i > x_j/c_j$ and $c_i, c_j > 0$.
Then it can be shown algebraically that the following two statements hold:

\begin{enumerate}
\item If $c_j \geq c_i$, then for sufficiently small $\varepsilon > 0$, $\frac{x_i}{c_i} + \frac{x_j}{c_j} > \frac{x_i}{c_i + \varepsilon} + \frac{x_j}{c_j - \varepsilon}$.
\item If $c_j < c_i$, then for sufficiently small $\varepsilon > 0$, $\frac{x_i}{c_i} + \frac{x_j}{c_j} > \frac{x_i}{c_i - \varepsilon} + \frac{x_j}{c_j + \varepsilon}$.
\end{enumerate}

Since the ratios of every pair of monomer types are equal,

$$\frac{c_i}{1} \leq \frac{c_i}{\sum_{i=1}^{k}{c_i}} = \frac{x_i}{\sum_{i=1}^{k}{x_i}} \leq \frac{x_i}{\log{n}}$$

So $\log{n} \leq x_i/c_i$ and $k\log{n} \leq \sum_{i=1}^{k}x_i/c_i$.
By Lemma~\ref{lem:usable-insertion-sites-ub}, since the insertion subsequence has length $\log(n)/2$ and no repeated insertion sites, $k \geq \log^{2/3}(n)/4$.
So the total expected time is $k\log{n} \geq \log^{2/3}(n)/8$.
\end{proof}

\section*{Acknowledgments}

The authors thank anonymous reviewers for comments that improved the readability and correctness of the paper.

\bibliographystyle{plain}
\bibliography{insertion_primitive}

\end{document}